\documentclass{scrartcl}
\pdfoutput=1
\usepackage{url}

\usepackage{appendix}
\usepackage{amsmath,amssymb,amsthm}
\usepackage{xspace}
\usepackage{bbm}
\usepackage{booktabs}
\usepackage{graphicx}
\usepackage{tikz}
\usepackage{pgfplots}
\usepackage{trsym}
\usepackage{siunitx}
\usepackage{cite}

\newcommand{\pconv}{\overset{p}{\to}}
\newcommand{\de}{\,\mathrm{d}}

\newcommand{\seta}{\ensuremath{\mathcal{A}}\xspace}
\newcommand{\setb}{\ensuremath{\mathcal{B}}\xspace}
\newcommand{\setc}{\ensuremath{\mathcal{C}}\xspace}

\newcommand{\seti}{\ensuremath{\mathcal{I}}\xspace}

\newcommand{\sets}{\ensuremath{\mathcal{S}}\xspace}
\newcommand{\sett}{\ensuremath{\mathcal{T}}\xspace}

\newcommand{\setx}{\ensuremath{\mathcal{X}}\xspace}
\newcommand{\sety}{\ensuremath{\mathcal{Y}}\xspace}
\DeclareMathOperator{\expop}{\mathbb{E}}
\DeclareMathOperator{\entop}{\mathbb{H}}

\DeclareMathOperator{\idop}{\mathbb{D}}
\DeclareMathOperator{\miop}{\mathbb{I}}
\DeclareMathOperator{\supp}{supp}
\DeclareMathOperator*{\argmax}{argmax}

\newcommand{\veca}{\boldsymbol{a}}
\newcommand{\vecb}{\boldsymbol{b}}
\newcommand{\vecB}{\boldsymbol{B}}

\newcommand{\vecx}{\boldsymbol{x}}
\newcommand{\vecX}{\boldsymbol{X}}
\newcommand{\vecy}{\boldsymbol{y}}
\newcommand{\vecY}{\boldsymbol{Y}}

\newcommand{\ogeq}[1]{\overset{\text{(#1)}}{\geq}}

\newcommand{\bpm}{\begin{pmatrix}}
\newcommand{\epm}{\end{pmatrix}}
\newcommand{\bbm}{\begin{bmatrix}}
\newcommand{\ebm}{\end{bmatrix}}

\theoremstyle{plain}

\newtheorem{lemma}{Lemma}
\newtheorem{proposition}{Proposition}
\theoremstyle{remark}
\newtheorem{remark}{Remark}

\theoremstyle{definition}
\usepackage{mdframed}
\definecolor{examplegray}{rgb}{.95,.95,.95}
\newtheorem{mdexample}{Example}[section]
\newenvironment{example}%
  {\begin{mdframed}[backgroundcolor=examplegray,hidealllines=true]\begin{mdexample}}%
  {\end{mdexample}\end{mdframed}}

\newcounter{problemcount}[section]

\newcommand{\rcode}{R_\textnormal{c}}
\newcommand{\rbcode}{R_\textnormal{b}}
\newcommand{\tildercode}{\tilde{R}_\textnormal{c}}
\newcommand{\rlm}{R_\textnormal{LM}}
\newcommand{\rtrans}{R_\textnormal{tx}}
\newcommand{\tcode}{T_\textnormal{c}}
\newcommand{\tbcode}{T_\textnormal{abc}}
\newcommand{\ttrans}{R_\textnormal{\sffamily ps}}
\newcommand{\hattcode}{\hat{T}_\textnormal{c}}

\newcommand{\xop}{\mathbb{X}}

\makeindex

\title{Achievable Rates for Probabilistic Shaping}

\author{Georg B\"ocherer\\
Mathematical and Algorithmic Sciences Lab\\
Huawei Technologies France S.A.S.U.\\
\url{georg.boecherer@ieee.org}}

\begin{document}

\maketitle

\begin{abstract}
For a layered probabilistic shaping (PS) scheme with a general decoding metric, an achievable rate is derived using Gallager's error exponent approach and the concept of achievable code rates is introduced. Several instances for specific decoding metrics are discussed, including bit-metric decoding, interleaved coded modulation, and hard-decision decoding. It is shown that important previously known achievable rates can also be achieved by layered PS. A practical instance of layered PS is the recently proposed probabilistic amplitude shaping (PAS).
\end{abstract}
\section{Introduction}
Communication channels often have non-uniform capacity-achieving input distributions, which is the main motivation for probabilistic shaping (PS), i.e., the development of practical transmission schemes that use non-uniform input distributions. Many different PS schemes have been proposed in literature, see, e.g., the literature review in \cite[Section~II]{bocherer2015bandwidth}.

In \cite{bocherer2015bandwidth}, we proposed probabilistic amplitude shaping (PAS), a layered PS architecture that concatenates a distribution matcher (DM) with a systematic encoder of a forward error correcting (FEC) code. In a nutshell, PAS works as follows. The DM serves as a shaping encoder and maps data bits to non-uniformly distributed (`shaped') amplitude sequences, which are then systematically FEC encoded, preserving the amplitude distribution. The additionally generated redundancy bits are mapped to sign sequences that are multiplied entrywise with the amplitude sequences, resulting in a capacity-achieving input distribution for the practically relevant discrete-input additive white Gaussian noise channel.

In this work, we take an information-theoretic perspective and use random coding arguments following Gallager's error exponent approach~\cite[Chapter~5]{gallager1968information} to derive achievable rates for a layered PS scheme of which PAS is a practical instance. Because rate and FEC code rate are different for layered PS, we introduce achievable \emph{code} rates.  The proposed achievable rate is amenable to analysis and we instantiate it for several special cases, including bit-metric decoding, interleaved coded modulation, hard-decision decoding, and binary hard-decision decoding.

Section~\ref{sec:preliminaries} provides preliminaries and notation. We define the layered PS scheme in Section~\ref{sec:layered ps}. In Section~\ref{sec:results}, we state and discuss the main results for a generic decoding metric. We discuss metric design and metric assessment in Section~\ref{sec:metric design} and Section~\ref{sec:metric assessment}, respectively. The main results are proven in Section~\ref{sec:proofs}.

\section{Preliminaries}
\label{sec:preliminaries}

\subsection{Empirical Distributions}
Let $\setx$ be a finite set and consider a length $n$ sequence $x^n=x_1x_2\dotsb x_n$ with entries $x_i\in\setx$. Let $N(a|x^n)$ be the number of times that letter $a\in\setx$ occurs in $x^n$, i.e.,
\begin{align}
N(a|x^n)=|\{i\colon x_i=a\}|.
\end{align}
The empirical distribution (type) of $x^n$ is
\begin{align}
P_X(a)=\frac{N(a|x^n)}{n},\qquad a\in\setx.
\end{align}
The type $P_X$ can also be interpreted as a probability distribution $P_X$ on $\setx$, assigning to each letter $a\in\setx$ the probability $\Pr(X=a)=P_X(a)$. The concept of letter-typical sequences as defined in \cite[Section~1.3]{kramer2007topics} describes a set of sequences that have approximately the same type. For $\epsilon\geq 0$, we say $x^n$ is \emph{$\epsilon$-letter-typical} with respect to $P_X$ if for each letter $a\in\setx$,
\begin{align}
(1-\epsilon)P_X(a)\leq\frac{N(a|x^n)}{n}\leq (1+\epsilon)P_X(a),\quad\forall a\in\setx.\label{eq:def:typ}
\end{align}
The sequences \eqref{eq:def:typ} are called typical in \cite[Section~3.3]{masseyapplied1},\cite[Section~2.4]{elgamal2011network} and robust typical in \cite[Appendix]{orlitsky2001coding}.
We denote the set of letter typical sequences by $\sett^n_\epsilon(P_X)$.
\subsection{Expectations}
For a real-valued function $f$ on $\setx$, the expectation of $f(X)$ is
\begin{align}
\expop[f(X)]=\sum_{a\in\supp P_X}P_X(a)f(a)
\end{align}
where $\supp P_X$ is the support of $P_X$. The conditional expectation is
\begin{align}
\expop[f(X)|Y=b]=\sum_{a\in\supp P_X}P_{X|Y}(a|b)f(a)
\end{align}
where for each $b\in\sety$, $P_{X|Y}(\cdot|b)$ is a distribution on $\setx$. Accordingly
\begin{align}
\expop[f(X)|Y]=\sum_{a\in\supp P_X}P_{X|Y}(a|Y)f(a)
\end{align}
is a random variable and
\begin{align}
\expop[\expop[f(X)|Y]]=\expop[f(X)].
\end{align}

\subsection{Information Measures}
Entropy of a discrete distribution $P_X$ is
\begin{align}
\entop(P_X)=\entop(X)=\expop[-\log_2 P_X(X)].
\end{align}
The conditional entropy is
\begin{align}
\entop(X|Y)=\expop[-\log_2 P_{X|Y}(X|Y)]
\end{align}
and the mutual information is
\begin{align}
  \miop(X;Y)=\entop(X)-\entop(X|Y).
\end{align}
The cross-entropy of two distributions $P_X,P_Z$ on $\setx$ is
\begin{align}
\xop(P_X\Vert P_Z)=\expop\left[-\log_2 P_Z(X)\right].\label{eq:cross entropy}
\end{align}
Note that the expectation in \eqref{eq:cross entropy} is taken with respect to $P_X$. The informational divergence of two distributions $P_X,P_Z$ on $\setx$ is
\begin{align}
\idop(P_X\Vert P_Z)&=\expop\left[\log_2\frac{P_X(X)}{P_Z(X)}\right]\\
&=\xop(P_X\Vert P_Z)-\entop(X).
\end{align}
We define the uniform distribution on $\setx$ as
\begin{align}
P_U(a)=\frac{1}{|\setx|},\quad a\in\setx.
\end{align}
We have
\begin{align}
  \idop(P_X\Vert P_U)=\entop(U)-\entop(X)=\log_2|\setx|-\entop(X).\label{eq:uniform divergence}
\end{align}
The information inequality states that
\begin{align}
\idop(P_X\Vert P_Z)\geq 0\label{eq:information inequality}
\end{align}
and equivalently
\begin{align}
\xop(P_X\Vert P_Z)\geq \entop(P_X)
\end{align}
with equality if and only if $P_X=P_Z$.

\section{Layered Probabilistic Shaping}
\label{sec:layered ps}
\begin{figure}[t]
\includegraphics{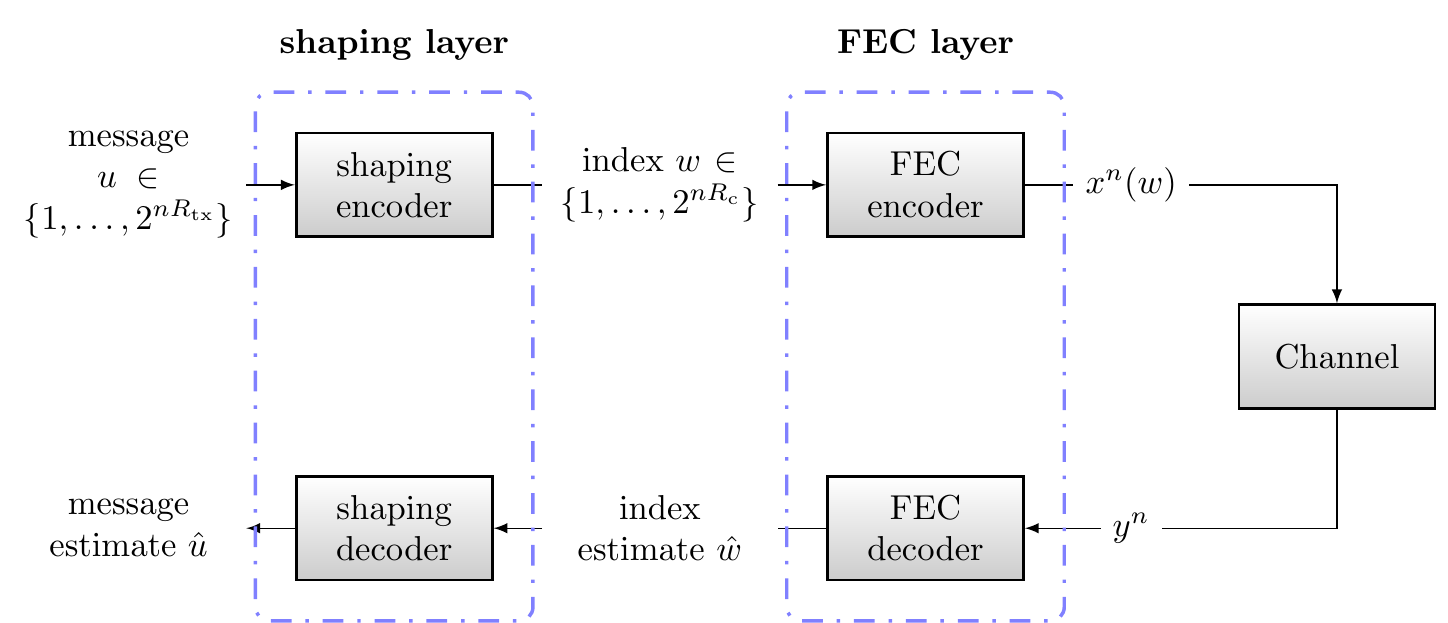}
\caption{Layered PS capturing the essence of PAS proposed in \cite{bocherer2015bandwidth}.}
\label{fig:layered ps}
\end{figure}
We consider the following transceiver setup (see also Figure~\ref{fig:layered ps}):
\begin{itemize}
    \item We consider a \emph{discrete-time channel} with input alphabet $\setx$ and output alphabet $\sety$. We derive our results assuming continuous-valued output. Our results also apply for discrete output alphabets.
    \item \emph{Random coding:} For indices $w=1,2,\dotsc,|\setc|$, we generate code words $C^n(w)$ with the $n|\setc|$ entries independent and uniformly distributed on $\setx$. The code is
    \begin{align}
        \setc=\{C^n(1),C^n(2),\dotsc,C^n(|\setc|)\}.
    \end{align}
    \item The \emph{code rate} is $\rcode=\frac{\log_2(|\setc|)}{n}$ and equivalently, we have $|\setc|=2^{n\rcode}$.
    \item \emph{Encoding:} We set $\rtrans+R'=\rcode$ and double index the code words by $C^n(u,v)$, $u=1,2,\dotsc,2^{n\rtrans}$, $v=1,2,\dotsc,2^{nR'}$. We encode message $u\in\{1,\dotsc,2^{n\rtrans}\}$ by looking for a $v$, so that $C^n(u,v)\in\sett^n_\epsilon(P_X)$. If we can find such $v$, we transmit the corresponding code word. If not, we choose some arbitrary $v$ and transmit the corresponding code word.
    \item The \emph{transmission rate} is $\rtrans$, since the encoder can encode $2^{n\rtrans}$ different messages.
    \item \emph{Decoding:} We consider a non-negative metric $q$ on $\setx\times\sety$ and we define
    \begin{align}
        q^n(x^n,y^n):=\prod_{i=1}^n q(x_i,y_i),\quad  x^n\in\setx^n,y^n\in\sety^n.\label{eq:airestim:metric}
    \end{align}
    For the channel output $y^n$, we let the receiver decode with the rule
    \begin{align}
        \hat{W}=\argmax_{w\in\{1,\dotsc,2^{n\rcode}\}}\prod_{i=1}^n q(C_i(w),y_i).
    \end{align}
    Note that the decoder evaluates the metric on all code words in $\setc$, which includes code words that will never be transmitted because they are not in the shaping set $\sett^n_\epsilon(P_X)$.
    \item \emph{Decoding error:} We consider the error probability
    \begin{align}
      P_e=\Pr(\hat{W}\neq W)
    \end{align}
    where $W$ is the index of the transmitted code word and $\hat{W}$ is the detected index at the receiver. Note that $\hat{W}=W$ implies $\hat{U}=U$, where $U$ is the encoded message and where $\hat{U}$ is the detected message. In particular, we have $\Pr(\hat{U}\neq U)\leq P_e$.
\end{itemize}
\begin{remark}\label{rem:classical}
  The classical transceiver setup analyzed in, e.g., \cite[Chapter~5~\&~7]{gallager1968information},\cite{kaplan1993information},\cite{ganti2000mismatched}, is as follows:
  \begin{itemize}
  \item \emph{Random coding:} For the code $\tilde{\setc}=\{\tilde{C}^n(1),\dotsc,\tilde{C}^n(2^{n\tildercode})\}$, the $n\cdot 2^{n\tildercode}$ code word entries are generated independently according to the distribution $P_X$.
  \item \emph{Encoding:} Message $u$ is mapped to code word $\tilde{C}^n(u)$.
  \item The decoder uses the decoding rule
  \begin{align}
  \hat{u}=\argmax_{u\in\{1,2,\dotsc,2^{n\tildercode}\}}\prod_{i=1}^n q(\tilde{C}_i(u),y_i).
  \end{align}
  \end{itemize}
  Note that in difference to layered PS, the code word index is equal to the message, i.e., $w=u$, and consequently, the transmission rate is equal to the code rate, i.e., $\rtrans=\tildercode$, while for layered PS, we have $\rtrans<\rcode$.
\end{remark}
\begin{remark}\label{rem:uniform}
  In case the input distribution $P_X$ is uniform, layered PS is equivalent to the classical transceiver.
\end{remark}

\section{Main Results}
\label{sec:results}
\subsection{Achievable Encoding Rate}

\begin{proposition}\label{prop:encoding rate}
Layered PS encoding is successful with high probability for large $n$ if
\begin{align}
\rtrans < [\rcode-\idop(P_X\Vert P_U)]^+.\label{eq:encoding rate}
\end{align}
\end{proposition}
\begin{proof}
See Section~\ref{subsec:encoding proof}.
\end{proof}
If the right-hand side of \eqref{eq:encoding rate} is positive, this condition means the following: out of the $2^{n\rcode}$ code words, approximately $2^{n[\rcode-\idop(P_X\Vert P_U)]}$ have approximately the distribution $P_X$ and may be selected by the encoder for transmission. If the code rate is less than the informational divergence, then very likely, the code does not contain any code word with approximately the distribution $P_X$. In this case, encoding is impossible, which corresponds to the encoding rate zero. The plus operator $[\cdot]^+=\max\{0,\cdot\}$ ensures that this is reflected by the expression on the right-hand side of \eqref{eq:encoding rate}.

\subsection{Achievable Decoding Rate}

\begin{proposition}
Suppose code word $C^n(w)=x^n$ is transmitted and let $y^n$ be a channel output sequence. With high probability for large $n$, the layered PS decoder can recover the index $w$ from the sequence $y^n$ if
\begin{align}
\rcode<\hattcode(x^n,y^n,q)=\log_2|\setx|-\frac{1}{n}\sum_{i=1}^n\left[-\log_2\frac{q(x_i,y_i)}{\sum_{a\in\setx}q(a,y_i)}\right]\label{eq:achievable decoding rate}
\end{align}
that is, $\hattcode(x^n,y^n,q)$ is an achievable code rate.
\end{proposition}
\begin{proof}
  See Section~\ref{subsec:decoding proof}.
\end{proof}
%The factor $1/|\setx|$ in \eqref{eq:achievable decoding rate} reflects that the code word entries are generated uniformly at random in the random coding experiment.
\begin{proposition}
For a memoryless channel with channel law
\begin{align}
p_{Y^n|X^n}(b^n|a^n)=\prod_{i=1}^n p_{Y|X}(b_i|a_i),\quad b^n\in\sety^n, a^n\in\setx^n
\end{align}
the layered PS decoder can recover sequence $x^n$ from the random channel output if the sequence is approximately of type $P_X$ and if
\begin{align}
\rcode<\tcode=\log_2|\setx|-\expop\left[-\log_2\frac{q(X,Y)}{\sum_{a\in\setx}q(a,Y)}\right]\label{eq:memoryless achievable decoding rate}
\end{align}
where the expectation is taken according to $XY\sim P_Xp_{Y|X}$.
\end{proposition}
\begin{proof}
  See Section~\ref{subsec:memoryless decoding proof}.
\end{proof}
The term 
\begin{align}
\expop\left[-\log_2\frac{q(X,Y)}{\sum_{a\in\setx}q(a,Y)}\right]\label{eq:unceratinty}
\end{align}
in \eqref{eq:memoryless achievable decoding rate} and its empirical version in \eqref{eq:achievable decoding rate} play a central role in achievable rate calculations. We call \eqref{eq:unceratinty} \emph{uncertainty}. Note that for each realization $b$ of $Y$, $Q_{X|Y}(\cdot|b):=q(\cdot,b)/\sum_{a\in\setx}q(a,b)$ is a distribution on $\setx$ so that 
\begin{align}
\expop\left[\left.-\log_2\frac{q(X,Y)}{\sum_{a\in\setx}q(a,Y)}\right|Y=b\right]&=\sum_{a\in\setx}P_{X|Y}(a|b)\log_2\left[-Q_{X|Y}(a|b)\right]\\
&=\xop\left(P_{X|Y}(\cdot|b)\Vert Q_{X|Y}(\cdot|b)\right)
\end{align}
is the cross-entropy of $P_{X|Y}(\cdot|b)$ and $Q_{X|Y}(\cdot|b)$. Thus, the uncertainty in \eqref{eq:unceratinty} is a conditional cross-entropy of the probabilistic model $Q_{X|Y}$ assumed by the decoder via its decoding metric $q$, and the actual distribution $P_Xp_{Y|X}$.
\begin{example}[Achievable Binary Code (ABC) Rate]
  For bit-metric decoding (BMD), which we discuss in detail in Section~\ref{subsec:bmd}, the input is a binary label $\vecB=B_1B_2\dotsc B_m$ and the optimal bit-metric is
  \begin{align}
    q(\veca,y)=\prod_{j=1}^m P_{B_j}(a_j)p_{Y|B_j}(y|a_j)
  \end{align}
  and the binary code rate is $\rbcode=\rcode/m$. The achievable binary code (ABC) rate is then
  \begin{align}
    \tbcode=\frac{\tcode}{m}=1-\frac{1}{m}\sum_{j=1}^m\entop(B_j|Y)
  \end{align}
  where we used $\log_2|\setx|=m$. We remark that ABC rates were used implicitly in \cite[Remark~6]{bocherer2015bandwidth} and \cite[Eq.~(23)]{steiner2016protograph} for the design of binary low-density parity-check (LDPC) codes.
\end{example}

\subsection{Achievable Transmission Rate}

By replacing the code rate $\rcode$ in the achievable encoding rate $[\rcode-\idop(P_X\Vert P_U)]^+$ by the achievable decoding rate $\tcode$, we arrive at an achievable transmission rate.
\begin{proposition}
  An achievable transmission rate is
\begin{align}
\ttrans = [\tcode-\idop(P_X\Vert P_U)]^+&= \left[\log_2|\setx|-\expop\left[-\log_2\frac{q(X,Y)}{\sum_{a\in\setx}q(a,Y)}\right]-\idop(P_X\Vert P_U)\right]^+\nonumber\\
&=\left[\entop(X)-\expop\left[-\log_2\frac{q(X,Y)}{\sum_{a\in\setx}q(a,Y)}\right]\right]^+\label{eq:uncertainty perspective}\\
&=\left[\expop\left[\log_2\frac{q(X,Y)}{\sum_{a\in\setx}\frac{1}{|\setx|}q(a,Y)}\right]-\idop(P_X\Vert P_U)\right]^+\label{eq:large code perspective}\\
&=\left[\expop\left[\log_2\frac{q(X,Y)\frac{1}{P_X(X)}}{\sum_{a\in\setx}q(a,Y)}\right]\right]^+\label{eq:output perspective}.
\end{align}
\end{proposition}
The right-hand sides provide three different perspectives on the achievable transmission rate.
\begin{itemize}
\item \emph{Uncertainty perspective:} In \eqref{eq:uncertainty perspective}, $q(\cdot,b)/\sum_{a\in\setx}q(a,b)$ defines for each realization $b$ of $Y$ a distribution on $\setx$ and plays the role of a posterior probability distribution that the receiver assumes about the input, given its output observation. The expectation corresponds to the uncertainty that the receiver has about the input, given the output.
\item \emph{Divergence perspective:} The term in \eqref{eq:large code perspective} emphasizes that the random code was generated according to a uniform distribution and that of the $2^{n\tcode}$ code words, only approximately $ 2^{n\tcode}/2^{n\idop(P_X\Vert P_U)}$ code words are actually used for transmission, because the other code words very likely do not have distributions that are approximately $P_X$.
\item \emph{Output perspective:} In \eqref{eq:output perspective}, $q(a,\cdot)/P_X(a)$ has the role of a channel likelihood given input $X=a$ assumed by the receiver, and correspondingly, $\sum_{a\in\setx}q(a,\cdot)$ plays the role of a channel output statistics assumed by the receiver.
\end{itemize}

\section{Metric Design: Examples}
\label{sec:metric design}

By the information inequality \eqref{eq:information inequality}, we know that
\begin{align}
  \expop[-\log_2 P_Z(X)]\geq \expop[-\log_2 P_X(X)]=\entop(X)\label{eq:entropy inequality}
\end{align}
with equality if and only if $P_Z=P_X$. We now use this observation to choose  optimal metrics.

\subsection{Mutual Information}

Suppose we have no restriction on the decoding metric $q$. To maximize the achievable rate, we need to minimize the uncertainty in \eqref{eq:uncertainty perspective}. We have
\begin{align}
\expop\left[-\log_2\frac{q(X,Y)}{\sum_{a\in\setx}q(a,Y)}\right]&=\expop\left[\expop\left[\left.-\log_2\frac{q(X,Y)}{\sum_{a\in\setx}q(a,Y)}\right\vert Y\right]\right]\\
&\ogeq{\ref{eq:entropy inequality}} \expop\left[\expop\left[\left.-\log_2P_{X|Y}(X|Y)\right\vert Y\right]\right]\\
&=\entop(X|Y)
\end{align}
with equality if we use the posterior probability distribution as metric, i.e.,
\begin{align}
q(a,b)=P_{X|Y}(a|b),\quad a\in\setx,b\in\sety.
\end{align}
Note that this choice of $q$ is not unique, in particular, $q(a,b)=P_{X|Y}(a|b)P_Y(b)$ is also optimal, since the factor $P_Y(b)$ cancels out. For the optimal metric, the achievable rate is
\begin{align}
\ttrans^\textsf{opt}=[\entop(X)-\entop(X|Y)]^+=\miop(X;Y)
\end{align}
where we dropped the $(\cdot)^+$ operator because by the information inequality, mutual information is non-negative.
\subsubsection*{Discussion}
In \cite[Chapter~5 ~\&~7]{gallager1968information}, the achievability of mutual information is shown using the classical transceiver of Remark~\ref{rem:classical} with the likelihood decoding metric $q(a,b)=p_{Y|X}(b|a)$, $a\in\setx,b\in\sety$. Comparing the classical transceiver with layered PS for a common rate $\rtrans$, we have
\begin{align}
  \text{classical transceiver: }&\hat{w}=\argmax_{w\in\{1,2,\dotsc,2^{n\rtrans}\}}\prod_{i=1}^np_{Y|X}(y_i|\tilde{c}_i(w))\label{eq:classical decoder}\\
  \text{layered PS: }&\hat{w}=\argmax_{w\in\{1,2,\dotsc,2^{n[\rtrans+\idop(P_X\Vert P_U)]}\}}\prod_{i=1}^nP _{X|Y}(c_i(w)|y_i)\nonumber\\
  &=\argmax_{w\in\{1,2,\dotsc,2^{n[\rtrans+\idop(P_X\Vert P_U)]}\}}\prod_{i=1}^n p_{Y|X}(y_i|c_i(w))P_X(c_i(w))\label{eq:ps decoder}
\end{align}
Comparing \eqref{eq:classical decoder} and \eqref{eq:ps decoder} suggests the following interpretation:
\begin{itemize}
  \item The classical transceiver uses the prior information by evaluating the \textbf{likelihood density} $p_{Y|X}$ on the code $\tilde{\setc}$ that contains code words with distribution $P_X$. The code $\tilde{\setc}$ has size $|\tilde{\setc}|=2^{n\rtrans}$.
  \item Layered PS uses the prior information by evaluating the \textbf{posterior distribution} on all code words in the `large' code $\setc$ that contains mainly code words that do not have distribution $P_X$. The code $\setc$ has size $|\setc|=2^{n[\rtrans+\idop(P_X\Vert P_U)]}$.
\end{itemize}
\begin{remark}
  The code $\tilde{\setc}$ of the classical transceiver is in general non-linear, since the set of vectors with distribution $P_X$ is non-linear. It can be shown that all the presented results for layered PS also apply when $\setc$ is a random linear code. In this case, layered PS evaluates a metric on a linear set while the classical transceiver evaluates a metric on a non-linear set.
\end{remark}

\subsection{Bit-Metric Decoding}
\label{subsec:bmd}
Suppose the channel input is a binary vector $\vecB=B_1\dotsb B_m$ and the receiver uses a bit-metric, i.e.,
\begin{align}
q(\veca,y)=\prod_{j=1}^m q_j(a_j,y).
\end{align}
In this case, we have for the uncertainty in \eqref{eq:uncertainty perspective}
\begin{align}
\expop\left[-\log_2\frac{q(\vecB,Y)}{\sum_{\veca\in\{0,1\}^m}q(\veca,Y)}\right]
&=\expop\left[-\log_2\frac{\prod_{j=1}^m q_j(B_j,Y)}{\sum_{\veca\in\{0,1\}^m}\prod_{j=1}^m q_j(a_j,Y)}\right]\\
&=\expop\left[-\log_2\frac{\prod_{j=1}^m q_j(B_j,Y)}{\prod_{j=1}^m \sum_{a\in\{0,1\}}q_j(a,Y)}\right]\\
&=\expop\left[-\sum_{j=1}^m\log_2\frac{q_j(B_j,Y)}{\sum_{a\in\{0,1\}}q_j(a,Y)}\right]\\
&=\sum_{j=1}^m\expop\left[-\log_2\frac{q_j(B_j,Y)}{\sum_{a\in\{0,1\}}q_j(a,Y)}\right].\label{eq:uc bmd}
\end{align}
For each $j=1,\dotsc,m$, we now have
\begin{align}
\expop\left[-\log_2\frac{q_j(B_j,Y)}{\sum_{a\in\{0,1\}}q_j(a,Y)}\right]&=\expop\left[\expop\left[\left.-\log_2\frac{q_j(B_j,Y)}{\sum_{a\in\{0,1\}}q_j(a,Y)}\right\vert Y\right]\right]\\
&\geq \entop(B_j|Y)
\end{align}
with equality if
\begin{align}
q_j(a,b)=P_{B_j|Y}(a|b),\quad a\in\{0,1\},b\in\sety.
\end{align}
The achievable rate becomes the bit-metric decoding (BMD) rate
\begin{align}
\ttrans^\textsf{bmd}=\left[\entop(\vecB)-\sum_{j=1}^m\entop(B_j|Y)\right]^+\label{eq:rbmd}
\end{align}
which we first stated in \cite{bocherer2014probabilistic} and discuss in detail in \cite[Section~VI.]{bocherer2015bandwidth}. In \cite{bocherer2016achievable}, we prove the achievability of \eqref{eq:rbmd} for discrete memoryless channels. For independent bit-level $B_1,B_2,\dotsc,B_m$, the BMD rate can be also be written in the form
\begin{align}
  \ttrans^\textsf{bmd,ind}=\sum_{j=1}^m\miop(B_j;Y).
\end{align}

\subsection{Interleaved Coded Modulation}
\label{subsec:interleaved}
Suppose we have a vector channel with input $\vecX=X_1\dotsb X_m$ with distribution $P_{\vecX}$ on the input alphabet $\setx^m$ and output $\vecY=Y_1\dotsb Y_m$ with distributions $P_{\vecY|\vecX}(\cdot|\veca)$, $\veca\in\setx^m$, on the output alphabet $\sety^m$. We consider the following situation:
\begin{itemize}
\item The $Y_i$ are potentially correlated, in particular, we may have $Y_1=Y_2=\dotsb=Y_m$.
\item Despite the potential correlation, the receiver uses a memoryless metric $q$ defined on $\setx\times\sety$, i.e., a vector input $\vecx$ and a vector output $\vecy$ are scored by
\begin{align}
q^m(\vecx,\vecy)=\prod_{i=1}^m q(x_i,y_i).
\end{align}
The reason for this decoding strategy may be an interleaver between encoder output and channel input that is reverted at the receiver but not known to the decoder. We therefore call this scenario \emph{interleaved coded modulation}.
\end{itemize}
Using the same approach as for bit-metric decoding, we have
\begin{align}
  \frac{1}{m}\expop\left[-\log_2\frac{\prod_{i=1}^m q(X_i,Y_i)}{\sum_{\veca\in\setx^m}\prod_{i=1}^m q(a_i,Y_i)}\right]
  &=  \frac{1}{m}\expop\left[-\log_2\frac{\prod_{i=1}^m q(X_i,Y_i)}{\prod_{i=1}^m\sum_{a\in\setx} q(a,Y_i)}\right]\\
  &=  \frac{1}{m}\sum_{i=1}^m \expop\left[-\log_2\frac{q(X_i,Y_i)}{\sum_{a\in\setx} q(a,Y_i)}\right].\label{eq:uc interleaved}
\end{align}
This expression is not very insightful. We could optimize $q$ for, say, the $i$th term, which would be
\begin{align}
q(a,b)=P_{X_i|Y_i}(a|b),\quad a\in\setx,b\in\sety
\end{align}
but this would not be optimal for the other terms. We therefore choose a different approach. Let $I$ be a random variable uniformly distributed on $\seti=\{1,2,\dotsc,m\}$ and define $X=X_I$, $Y=Y_I$. Then, we have
\begin{align}
\frac{1}{m}\sum_{i=1}^m \expop\left[-\log_2\frac{q(X_i,Y_i)}{\sum_{a\in\setx} q(a,Y_i)}\right]=\expop\left[-\log_2\frac{q(X_I,Y_I)}{\sum_{a\in\setx} q(a,Y_I)}\right]\\
=\expop\left[-\log_2\frac{q(X,Y)}{\sum_{a\in\setx} q(a,Y)}\right].\label{eq:uc interleaving}
\end{align}
Thus, the optimal metric for interleaving is
\begin{align}
q(a,b)=P_{X|Y}(a|b)
\end{align}
which can be calculated from
\begin{align}
P_{X}(a)p_{Y|X}(b|a)=\sum_{j=1}^m \frac{1}{m}P_{X_j}(a)p_{Y_j|X_j}(b|a).
\end{align}
The achievable rate becomes
\begin{align}
\ttrans^\textsf{icm}=\left[\entop(\vecX)-m\entop(X|Y)\right]^+.
\end{align}
% A special case is when the $X_i$ are binary, in this case, the achievable rate is
%   \begin{align}
%   \ttrans^\textsf{int,bin}=\left[\entop(\vecB)-m\entop(B|Y)\right]^+.
%   \end{align}

  \section{Metric Assessment: Examples}
  \label{sec:metric assessment}

  Suppose a decoder is constrained to use a specific metric $q$. In this case, our task is to assess the metric performance by calculating a rate that can be achieved by using metric $q$. If $q$ is a non-negative metric, an achievable rate is our transmission rate expression
  \begin{align}
\ttrans(q)=\left[\entop(X)-\expop\left[-\log_2\frac{q(X,Y)}{\sum_{a\in\setx} q(a,Y)}\right]\right]^+.
  \end{align}
However, higher rates may also be achievable by $q$. The reason for this is as follows: suppose we have another metric $\tilde{q}$ that scores the code words in the same order as metric $q$, i.e., we have
\begin{align}
\tilde{q}(a_1,b)>\tilde{q}(a_2,b)\Leftrightarrow q(a_1,b)> q(a_2,b),\quad a_1,a_2\in\setx,b\in\sety.\label{eq:order preserving}
\end{align}
Then, $\ttrans(\tilde{q})$ is also achievable by $q$. An example for an order preserving transformation is $\tilde{q}(a,b)=e^{q(a,b)}$. For a non-negative metric $q$, another order preserving transformation is $\tilde{q}(a,b)=q(a,b)^s$ for $s>0$. We may now find a better achievable rate for metric $q$ by calculating for instance
\begin{align}
\max_{s>0}\ttrans(q^s).
\end{align}
In the following, we will say that two metrics $q$ and $\tilde{q}$ are equivalent if and only if the order-preserving condition \eqref{eq:order preserving} is fulfilled.

\subsection{Generalized Mutual Information}

Suppose the input distribution is uniform, i.e., $P_X(a)=1/|\setx|,a\in\setx$. In this case, we have
\begin{align}
\max_{s>0}\ttrans(q^s)&=\max_{s>0}\left[\expop\left[\log_2\frac{q(X,Y)^s\frac{1}{P_X(X)}}{\sum_{a\in\setx}q(a,Y)^s}\right]\right]^+\label{eq:almost gmi}\\
&=\max_{s>0}\expop\left[\log_2\frac{q(X,Y)^s}{\sum_{a\in\setx}P_X(a)q(a,Y)^s}\right]\label{eq:gmi}
\end{align}
where we used the output perspective \eqref{eq:output perspective} in \eqref{eq:almost gmi}, where we could move $P_X(a)$ under the sum in \eqref{eq:gmi}, because $P_X$ is by assumption uniform, and where we could drop the $(\cdot)^+$ operator because for $s=0$, the expectation is zero. The expression in \eqref{eq:gmi} is called generalized mutual information (GMI) in \cite{kaplan1993information} and was shown to be an achievable rate for the classical transceiver. This is in line with Remark~\ref{rem:uniform}, namely that for uniform input, layered PS is equivalent to the classical transceiver. For non-uniform input, the GMI and \eqref{eq:almost gmi} differ, i.e., we do not have equality in \eqref{eq:gmi}.
\subsubsection*{Discussion}

Suppose for a non-uniform input distribution $P_X$ and a metric $q$, the GMI evaluates to $R$, implying that a classical transceiver can achieve $R$. Can also layered PS achieve $R$, possibly by using a different metric? The answer is yes. Define
\begin{align}
  \tilde{q}(a,b)=q(a,b)P_X(a)^\frac{1}{s},\quad a\in\setx,b\in\sety\label{eq:screw metric}
\end{align}
where $s$ is the optimal value maximizing the GMI. We calculate a PS achievable rate for $\tilde{q}$ by analyzing the equivalent metric $\tilde{q}^{s}$. We have
\begin{align}
  \ttrans&=\left[\expop\left[\log_2\frac{\tilde{q}^{s}(X,Y)\frac{1}{P_X(X)}}{\sum_{a\in\setx}\tilde{q}^{s}(a,Y)}\right]\right]^+\\
&=\left[\expop\left[\log_2\frac{q^{s}(X,Y)}{\sum_{a\in\setx}P_X(a)q^{s}(a,Y)}\right]\right]^+\\
&=R
\end{align}
which shows that $R$ can also be achieved by layered PS. It is important to stress that this requires a change of the metric: for example, suppose $q$ is the Hamming metric of a hard-decision decoder (see Section~\ref{subsec:hd}). In general, this does \emph{not} imply that also $\tilde{q}$ defined by \eqref{eq:screw metric} is a Hamming metric.
\subsection{LM-Rate}

For the classical transceiver of Remark~\ref{rem:classical}, the work \cite{ganti2000mismatched} shows that the so-called LM-Rate defined as
\begin{align}
\rlm(s,r)=\left[\expop\left[\log_2\frac{q(X,Y)^s r(X)}{\sum_{a\in\supp P_X}P_X(a)q(a,Y)^s r(a)}\right]\right]^+
\end{align}
is achievable, where $s>0$ and where $r$ is a function on $\setx$. By choosing $s=1$ and $r(a)=1/P_X(a)$, we have
\begin{align}
\rlm(1,1/P_X)&=\left[\expop\left[\log_2\frac{q(X,Y)\frac{1}{P_X(X)} }{\sum_{a\in\supp P_X}q(a,Y)}\right]\right]^+\\
&\geq\left[\expop\left[\log_2\frac{q(X,Y)\frac{1}{P_X(X)} }{\sum_{a\in\setx}q(a,Y)}\right]\right]^+
\label{eq:lm equal}\\
&=\ttrans
\end{align}
with equality in \eqref{eq:lm equal} if $\supp P_X=\setx$. Thus, formally, our achievable transmission rate can be recovered from the LM-Rate. We emphasize that \cite{ganti2000mismatched} shows the achievability of the LM-Rate for the classical transceiver of Remark~\ref{rem:classical}, and consequently, $\rlm$ and $\ttrans$ have different operational meanings, corresponding to achievable rates of two different transceiver setups, with different random coding experiments, and different encoding and decoding strategies.

\subsection{Hard-Decision Decoding}
\label{subsec:hd}
\begin{figure}
  \centering
  \includegraphics{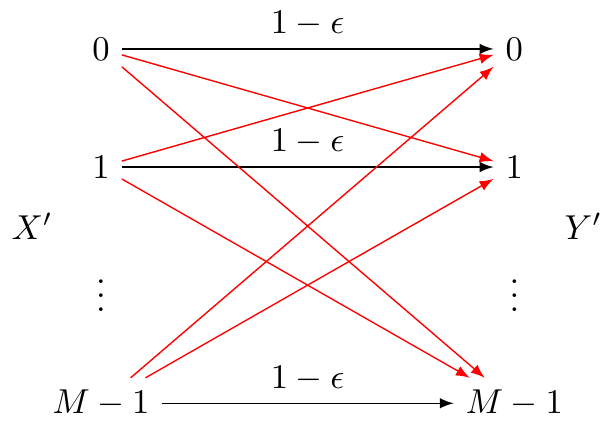}
  \caption{The $M$-ary symmetric channel. Each red transition has probability $\frac{\epsilon}{M-1}$. Note that for $M=2$, the channel is the binary symmetric channel. For uniformly distributed input $X'$, we have $\entop(X'|Y')=\entop_2(\epsilon)+\epsilon\log_2(M-1)$.}
  \label{fig:qary}
\end{figure}
Hard-decision decoding consists of two steps. First, the channel output alphabet is partitioned into disjoint decision regions
\begin{align}
  \sety=\bigcup_{a\in\setx}\sety_a,\quad \sety_a\cap\sety_b=\emptyset\text{ if }a\neq b
\end{align}
and a quantizer $\omega$ maps the channel output to the channel input alphabet according to the decision regions, i.e.,
\begin{align}
  \omega\colon \sety\to\setx,\quad\omega(b)=a\Leftrightarrow b\in\sety_a.
\end{align}
Second, the receiver uses the Hamming metric of $\setx$ for decoding, i.e.,
\begin{align}
  q(a,\omega(y))=\mathbbm{1}(a,\omega(y))=\begin{cases}
1,&\text{if }a=\omega(y)\\
0,&\text{otherwise}.
\end{cases}
\end{align}
We next derive an achievable rate by analyzing the equivalent metric $e^{s\mathbbm{1}(\cdot,\cdot)}$, $s>0$. For the uncertainty, we have
\begin{align}
  &\expop\left[-\log_2\frac{e^{s\mathbbm{1}[X,\omega(Y)]}}{\sum_{a\in\setx}e^{s\mathbbm{1}[a,\omega(Y)]}}\right]=  \expop\left[-\log_2\frac{e^{s\mathbbm{1}[X,\omega(Y)]}}{|\setx|-1+e^s}\right]\\
  &=  -\Pr[X=\omega(Y)]\log_2\frac{e^s}{|\setx|-1+e^s}-\Pr[X\neq\omega(Y)]\log_2\frac{1}{|\setx|-1+e^s}\\
  &=  -(1-\epsilon)\log_2\frac{e^s}{|\setx|-1+e^s}-\epsilon\log_2\frac{1}{|\setx|-1+e^s}\\
  &=  -(1-\epsilon)\log_2\frac{e^s}{|\setx|-1+e^s}-\sum_{\ell=1}^{|\setx|-1}\frac{\epsilon}{|\setx|-1}\log_2\frac{1}{|\setx|-1+e^s}
\end{align}
where we defined $\epsilon=\Pr(X\neq \omega(Y))$. By \eqref{eq:entropy inequality}, the last line is maximized by choosing
\begin{align}
s\colon 1-\epsilon=\frac{e^s}{|\setx|-1+e^s}\text{ and }\frac{\epsilon}{|\setx|-1}=\frac{1}{|\setx|-1+e^s}
\end{align}
which is achieved by
\begin{align}
  e^s=\frac{(|X|-1)(1-\epsilon)}{\epsilon}.
\end{align}
With this choice for $s$, we have
\begin{align}
  &-(1-\epsilon)\log_2(1-\epsilon)-\sum_{\ell=1}^{|\setx|-1}\frac{\epsilon}{|\setx|-1}\log_2\frac{\epsilon}{|\setx|-1}\nonumber\\
  &\qquad=  \underbrace{-(1-\epsilon)\log_2(1-\epsilon)-\epsilon\log_2\epsilon}_{=:\entop_2(\epsilon)} +\epsilon\log_2(|\setx|-1)\\
  &\qquad= \entop_2(\epsilon) +\epsilon\log_2(|\setx|-1)\label{eq:mary uc}
\end{align}
where $\entop_2(\cdot)$ is the binary entropy function. The term \eqref{eq:mary uc} corresponds to the conditional entropy of a $|\setx|$-ary symmetric channel with uniform input, see Figure~\ref{fig:qary} for an illustration. We conclude that by hard-decision decoding, we can achieve
\begin{align}
  \ttrans^\textsf{hd} = \left[\entop(X)-[\entop_2(\epsilon) +\epsilon\log_2(|\setx|-1)]\right]^+
\end{align}
where
\begin{align}
  \epsilon &= 1-\Pr[X=\omega(Y)]\\
  &=1-\sum_{a\in\setx}P_X(a)\int_{\sety_a}p_{Y|X}(\tau|a)\de \tau.
\end{align}
\subsection{Binary Hard-Decision Decoding}

Suppose the channel input is the binary vector $\vecB=B_1\dotsb B_m$ and the decoder uses $m$ binary quantizers, i.e., we have
\begin{align}
  \sety = \sety_{0j}\cup\sety_{1j},\quad \sety_{1j}=\sety\setminus\sety_{0j}\\
  \omega_j\colon \sety\to\{0,1\},\quad \omega_j(b)=a\Leftrightarrow b\in\sety_{ja}.
\end{align}
The receiver uses a binary Hamming metric, i.e.,
\begin{align}
  q(a,b)&=\mathbbm{1}(a,b),\quad a,b\in\{0,1\}\\
  q^m(\veca,\vecb)&=\sum_{j=1}^m \mathbbm{1}(a_j,b_j)
\end{align}
and we analyze the equivalent metric
\begin{align}
  e^{sq^m(\veca,\vecb)}=\prod_{j=1}^m e^{s\mathbbm{1}(a_j,b_j)},\quad s>0.
\end{align}
Since the decoder uses the same metric for each bit-level $j=1,2,\dotsc,m$, binary hard-decision decoding is an instance of interleaved coded modulation, which we discussed in Section~\ref{subsec:interleaved}. Thus, defining the auxiliary random variable $I$ uniformly distributed on $\{1,2,\dotsc,m\}$ and
\begin{align}
  B&=B_I,\quad \hat{B}=\omega_I(Y)\label{eq:average hd}
\end{align}
we can use the interleaved coded modulation result \eqref{eq:uc interleaving}. We have for the normalized uncertainty
\begin{align}
&\hspace{-2cm}\frac{1}{m}\expop\left[-\log_2\frac{\prod_{j=1}^m e^{s\mathbbm{1}[B_j,\omega_j(Y)]}}{\sum_{\veca\in\{0,1\}^m}\prod_{j=1}^m e^{s\mathbbm{1}[a_j,\omega_j(Y)]}}\right]\nonumber\\
&\overset{\text{\eqref{eq:uc interleaving},\eqref{eq:average hd}}}{=}  \expop\left[-\log_2\frac{e^{s\mathbbm{1}(B,\hat{B})}}{\sum_{a\in\{0,1\}}e^{s\mathbbm{1}(a,\hat{B})}}\right]\\
&=-\Pr(B=\hat{B})\log_2\frac{e^s}{e^s+1}-\underbrace{\Pr(B\neq \hat{B})}_{=:\epsilon}\log_2\frac{1}{e^s+1}\\
  &\ogeq{\ref{eq:entropy inequality}} \entop(\epsilon)
\end{align}
with equality if
\begin{align}
  s\colon \frac{1}{e^s+1}=\epsilon.
\end{align}
Thus, with a hard decision decoder, we can achieve
\begin{align}
\ttrans^\textsf{hd,bin}=\left[\entop(\vecB)-m\entop_2(\epsilon)\right]^+
\end{align}
where
\begin{align}
  \epsilon=\sum_{j=1}^m\frac{1}{m}\sum_{a\in\{0,1\}}P_{B_j}(a)\int_{\sety_{ja}}p_{Y|B_j}(\tau|a)\de\tau.
\end{align}
For uniform input, the rate becomes
\begin{align}
  R_\textsf{uni}^\textsf{hd,bin}=m-m\entop_2(\epsilon)=m[1-\entop_2(\epsilon)].
\end{align}

  \section{Proofs}
\label{sec:proofs}
  \subsection{Achievable Encoding Rate}
  \label{subsec:encoding proof}
  We consider a general shaping set $\sets_n\subseteq\setx^n$, of which $\sett_\epsilon^n(P_X)$ is an instance. An encoding error happens when for message $u\in\{1,2,\dotsc,2^{n\rtrans}\}$ with
  \begin{align}
  \rtrans\geq 0\label{eq:positive rate}
  \end{align}
    there is no $v\in\{1,2,\dotsc,2^{nR'}\}$ such that $C^n(u,v)\in\sets_n$, where $R'=\rcode-\rtrans$. In our random coding experiment, each code word is chosen uniformly at random from $\setx^n$ and it is not in $\sets_n$ with probability
  \begin{align}
  \frac{|\setx|^n-|\sets_n|}{|\setx|^n}=1-\frac{|\sets_n|}{|\setx|^n}.
  \end{align}
The probability that none of $2^{nR'}$ code words is in $\sets_n$ is therefore
\begin{align}
\left(1-\frac{|\sets_n|}{|\setx|^n}\right)^{2^{nR'}}\leq \exp\left(-\frac{|\sets_n|}{|\setx|^n}2^{nR'}\right)
\end{align}
where we used $1-s\leq \exp(-s)$.
We have
\begin{align}
\frac{|\sets_n|}{|\setx|^n}2^{nR'}=2^{n(\rcode-\rtrans+\frac{1}{n}\log_2|\sets_n|-\log_2|\setx|)}.
\end{align}
Thus, for
\begin{align}
\rtrans<\rcode-\left(\log_2|\setx|-\frac{1}{n}\log_2|\sets_n|\right)\label{eq:general achievable encoding rate}
\end{align}
the encoding error probability decays doubly exponentially fast with $n$. This bound can now be instantiated for specific shaping sets. Here, we consider the typical set $\sett^n_\epsilon(P_X)$ as defined in \eqref{eq:def:typ}. The exponential growth with $n$ of $\sett^n_\epsilon(P_X)$ is as follows.
  \begin{lemma}[{Typicality, \cite[Theorem~1.1]{kramer2007topics},\cite[Lemma~19]{orlitsky2001coding}}]
  Suppose $0<\epsilon<\mu_X$. We have
  \begin{align}
  (1-\delta_\epsilon(n,P_X))2^{n(1-\epsilon)\entop(X)}\leq |\sett_\epsilon^n(P_X)|\label{eq:typ 2}
  \end{align}
  where $\delta_\epsilon(P_X,n)$ is such that $\delta_\epsilon(P_X,n)\xrightarrow{n\to\infty}0$ exponentially fast in $n$.
  \end{lemma}
Thus,
\begin{align}
\frac{1}{n}\log_2|\sett^n_\epsilon(P_X)|\xrightarrow{n\to\infty}\entop(X).\label{eq:typical size limit}
\end{align}
For $\sets_n=\sett^n_\epsilon(P_X)$ in \eqref{eq:general achievable encoding rate} and by \eqref{eq:typical size limit}, if
\begin{align}
\rtrans&<\rcode-\left[\log_2|\setx|-\entop(X)\right]\\
&=\rcode-\idop(P_X\Vert P_U)\label{eq:divergence size}
\end{align}
then the encoding error probability can be made arbitrarily small by choosing $n$ large enough. Equality in \eqref{eq:divergence size} follows by \eqref{eq:uniform divergence}. By \eqref{eq:positive rate} and \eqref{eq:divergence size}, an achievable encoding rate is
\begin{align}
\rtrans=\left[\rcode-\idop(P_X\Vert P_U)\right]^+
\end{align}
which is the statement of Proposition~\ref{prop:encoding rate}.

\subsection{Achievable Code Rate}
\label{subsec:decoding proof}

\begin{figure}[t]
    \footnotesize
    \centering
    \includegraphics{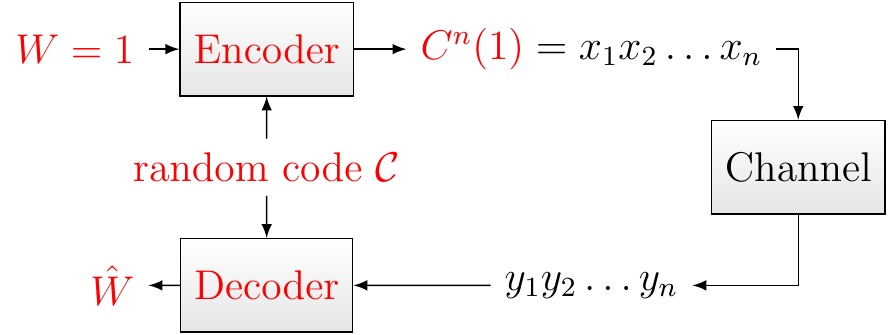}
    \caption{The considered setup.}
    \label{fig:coding experiment}
\end{figure}
We consider the setup in Figure~\ref{fig:coding experiment}, i.e., we condition on that index $W$ was encoded to $C^n(W)=x^n$ and that sequence $y^n$ was output by the channel. For notational convenience, we assume without loss of generality $W=1$. We have the implications
\begin{align}
    \hat{W}\neq 1&\Rightarrow \hat{W}=w'\neq 1\\
    &\Rightarrow L(w'):=\frac{q^n(C^n(w'),y^n)}{q^n(x^n,y^n)}\geq 1\\
    &\Rightarrow \sum_{w=2}^{|\setc|}L(w)\geq 1.
\end{align}
If event $\seta$ implies event $\setb$, then $\Pr(\seta)\leq\Pr(\setb)$. Therefore, we have
\begin{align}
    &\hspace{-2cm}\Pr(\hat{W}\neq 1|X^n=x^n,Y^n=y^n)\leq \Pr\left[\left.\sum_{w=2}^{|\setc|}L(w)\geq 1\right\vert X^n=x^n,Y^n=y^n\right]\\
    &\leq \expop\left[\left.\sum_{w=2}^{|\setc|}L(w)\right\vert X^n=x^n,Y^n=y^n\right]\label{eq:airestim:markov}\\
    &=q^n(x^n,y^n)^{-1}\expop\left[\sum_{w=2}^{|\setc|}q^n(C^n(w),y^n)\right]\label{eq:airestim:pairwise independence}\\
    &= (|\setc|-1)q^n(x^n,y^n)^{-1}\expop\left[q^n(C^n,y^n)\right]\label{eq:airestim:iid random coding}\\
    &\leq |\setc|q^n(x^n,y^n)^{-1}\expop\left[q^n(C^n,y^n)\right]\\
    &= |\setc|\frac{1}{\prod_{i=1}^n q(x_i,y_i)}\prod_{i=1}^n\expop\left[q(C,y_i)\right]\label{eq:airestim:memoryless metric}\\
    &= |\setc|\frac{1}{\prod_{i=1}^n q(x_i,y_i)}\prod_{i=1}^n\sum_{a\in\setx}|\setx|^{-1}q(a,y_i)\\
    &= |\setc|\prod_{i=1}^n\frac{\sum_{a\in\setx}q(a,y_i)}{q(x_i,y_i)|\setx|}
\end{align}
where
\begin{itemize}
    \item Inequality in \eqref{eq:airestim:markov} follows by Markov's inequality~\cite[Section~1.6.1]{gallager2013stochastic}.
    \item Equality in \eqref{eq:airestim:pairwise independence} follows because for $w\neq 1$, the code word $C^n(w)$ and the transmitted code word $C^n(1)$ were generated independently so that $C^n(w)$ and $[C^n(1),Y^n]$ are independent.
    \item Equality in \eqref{eq:airestim:iid random coding} holds because in our random coding experiment, for each index $w$, we generated the code word entries $C_1(w),C_2(w),\dotsc,C_n(w)$ iid.
    \item In \eqref{eq:airestim:memoryless metric}, we used \eqref{eq:airestim:metric}, i.e., that $q^n$ defines a memoryless metric.
\end{itemize}
We can now write this as
\begin{align}
    \Pr(\hat{W}\neq 1|X^n=x^n,Y^n=y^n)&\leq 2^{-n[\hattcode(x^n,y^n,q)-\rcode]}\\
        \text{where }\hattcode(x^n,y^n,q)&=\frac{1}{n}\sum_{i=1}^n\log_2\frac{q(x_i,y_i)}{\sum_{a\in\setx}\frac{1}{|\setx|}q(a,y_i)}\label{airestim:eq:code rate estimate}\\
        \rcode &= \frac{\log_2|\setc|}{n}.
\end{align}
For large $n$, the error probability upper bound is vanishingly small, if
\begin{align}
    \rcode < \hattcode(x^n,y^n,q).\label{airestim:eq:code rate codition}
\end{align}
Thus, $\hattcode(x^n,y^n,q)$ is an \emph{achievable code rate}\index{achievable code rate}, i.e., for a random code $\setc$, if \eqref{airestim:eq:code rate codition} holds, then, with high probability, sequence $x^n$ can be decoded from $y^n$.

\subsection{Achievable Code Rate for Memoryless Channels}
\label{subsec:memoryless decoding proof}
Consider now a \emph{memoryless channel}
\begin{align}
    p_{Y^n|X^n}&=\prod_{i=1}^n p_{Y_i|X_i}\\
    p_{Y_i|X_i}&=p_{Y|X},\quad i=1,2,\dotsc,n.
\end{align}
We continue to assume input sequence $x^n$ was transmitted, but we replace the specific channel output measurement $y^n$ by the random output $Y^n$, distributed according to $p_{Y|X}^n(\cdot|x^n)$. The achievable code rate \eqref{airestim:eq:code rate estimate} evaluated in $Y^n$ is
\begin{align}
    \hattcode(x^n,Y^n,q)=\frac{1}{n}\sum_{i=1}^n\log_2\frac{q(x_i,Y_i)}{\sum_{c\in\setx}\frac{1}{|\setx|}q(c,Y_i)}.\label{airestim:eq:code rate estimate random}
\end{align}
Since $Y^n$ is random, $\hattcode(x^n,Y^n,q)$ is also random. First, we rewrite \eqref{airestim:eq:code rate estimate random} by sorting the summands by the input symbols, i.e.,
\begin{align}
\frac{1}{n}\sum_{i=1}^n\log_2\frac{q(x_i,Y_i)}{\sum_{c\in\setx}\frac{1}{|\setx|}q(c,Y_i)}
=\sum_{a\in\setx}\frac{N(a|x^n)}{n}\left[\frac{1}{N(a|x^n)}\sum_{i\colon x_i=a}\log_2\frac{q(a,Y_i)}{\sum_{c\in\setx}\frac{1}{|\setx|}q(c,Y_i)}\right]\label{airestim:eq:code rate estimate2}
\end{align}
Note that identity \eqref{airestim:eq:code rate estimate2} holds also when the channel has memory. For memoryless channels, we make the following two observations:
\begin{itemize}
    \item Consider the inner sums in \eqref{airestim:eq:code rate estimate2}. For memoryless channels, the outputs $\{Y_i\colon x_i=a\}$ are iid according to $p_{Y|X}(\cdot|a)$. Therefore, by the Weak Law of Large Number~\cite[Section~1.7]{gallager2013stochastic},
    \begin{align}
        \frac{1}{N(a|x^n)}\sum_{i\colon x_i=a}\log_2\frac{q(a,Y_i)}{\sum_{c\in\setx}\frac{1}{|\setx|}q(c,Y_i)}\pconv\expop\left[\left.\log_2\frac{q(a,Y)}{\sum_{c\in\setx}\frac{1}{|\setx|}q(c,Y)}\right|X=a\right]\label{airestim:eq:memoryless proof2}
    \end{align}
    where $\pconv$ denotes convergence in probability~\cite[Section~1.7]{gallager2013stochastic}. That is, by making $n$ and thereby $N(a|x^n)$ large, each inner sum converges in probability to a deterministic value. Note that the expected value on the right-hand side of \eqref{airestim:eq:memoryless proof2} is no longer a function of the output sequence $Y^n$ and is determined by the channel law $p_{Y|X}(\cdot|a)$ according to which the expectation is calculated.
    \item Suppose now for some distribution $P_X$ and $\epsilon\geq 0$, we have $x^n\in\sett^n_\epsilon(P_X)$, in particular,
    \begin{align}
      \frac{N(a|x^n)}{n}\geq (1-\epsilon)P_X(a),\quad a\in\setx.
    \end{align}
    We now have
    \begin{align}
        &\hspace{-1cm}\frac{1}{n}\sum_{i=1}^n\log_2\frac{q(x_i,Y_i)}{\sum_{c\in\setx}\frac{1}{|\setx|}q(c,Y_i)}\nonumber\\
        &\pconv\sum_{a\in\setx}\frac{N(a|x^n)}{n}\expop\left[\left.\log_2\frac{q(a,Y)}{\sum_{c\in\setx}\frac{1}{|\setx|}q(c,Y)}\right|X=a\right]\\
&\geq\expop\left[\log_2\frac{q(X,Y)}{\sum_{c\in\setx}\frac{1}{|\setx|}q(c,Y)}\right]\nonumber\\
&\hspace{1cm}-\epsilon\sum_{a\in\setx}P_X(a)\left|\expop\left[\left.\log_2\frac{q(a,Y)}{\sum_{c\in\setx}\frac{1}{|\setx|}q(c,Y)}\right|X=a\right]\right|\label{airestim:eq:code rate estimate convergence}
          \end{align}
where the expectation in \eqref{airestim:eq:code rate estimate convergence} is calculated according to $P_X$ and the channel law $p_{Y|X}$. In other words, \eqref{airestim:eq:code rate estimate convergence} is an achievable code rate for all code words $x^n$ that are in the shaping set $\sett^n_\epsilon(P_X)$. \end{itemize}

\section*{Acknowledgment}

The author is grateful to Gianluigi Liva and Fabian Steiner for fruitful discussions encouraging this work. The author thanks Fabian Steiner for helpful comments on drafts. A part of this work was done while the author was with the Institute for Communications Engineering, Technical University of Munich, Germany.

%\printbibliography

\bibliographystyle{IEEEtran}
\bibliography{IEEEabrv,confs-jrnls-v1-DONOTEDIT,references-v5-DONOTEDIT,references}

\end{document}